\documentclass{article}
\usepackage{tikz}
\usepackage[english]{babel}
\usepackage[utf8x]{inputenc}
\usepackage{amsmath}
\usepackage{amsthm}
\usepackage{amsfonts}
\usepackage{graphicx}
\usepackage[colorinlistoftodos]{todonotes}
\usepackage{wrapfig}
\usepackage{enumerate}
\newtheorem{theorem}{Theorem}
\newtheorem{lemma}{Lemma}

\newtheorem{corollary}{Corollary}
\newtheorem{definition}{Definition}

\newcommand{\norm}[1]{\left \Vert #1\right\Vert}

\newcommand{\abs}[1]{\left\vert#1\right\vert}

\newcommand{\set}[1]{\left\{#1\right\}}
\newcommand{\parr}[1]{\left (#1\right )}
\newcommand{\brac}[1]{\left [#1\right ]}

\newcommand{\ip}[1]{\left \langle #1 \right \rangle }

\newcommand{\Real}{\mathbb R}
\newcommand{\Natural}{\mathbb N}
\newcommand{\Integer}{\mathbb Z}
\newcommand{\eps}{\varepsilon}

\newcommand{\too}{\rightarrow}

\newcommand{\T}{\mathcal{T}} 
\newcommand{\M}{\mathcal{M}}
\newcommand{\V}{\mathcal{V}}
\newcommand{\E}{\mathcal{E}}
\newcommand{\F}{\mathcal{F}}
\newcommand{\C}{\mathcal{C}} 
\newcommand{\U}{\mathcal{U}} 
\newcommand{\W}{\mathcal{W}} 
\newcommand{\B}{\mathcal{B}} 
\newcommand{\Q}{\mathcal{Q}} 

\newcommand{\area}[1]{\abs{#1}}
\newcommand{\length}[1]{\abs{#1}}

\newcommand {\1}{\mathrm{\textbf{1}}}
\newcommand{\one}{\1}

\newcommand{\eg}{{\it e.g.}}
\newcommand{\ie}{{\it i.e.}}

\title{A Linear Variational Principle for Riemann Mappings \\ and Discrete Conformality}
\author{Nadav Dym, Yaron Lipman, Raz Slutsky\\ Weizmann Institute of Science}
\begin{document}

        \maketitle
        
\begin{figure}
	\includegraphics[width=\columnwidth]{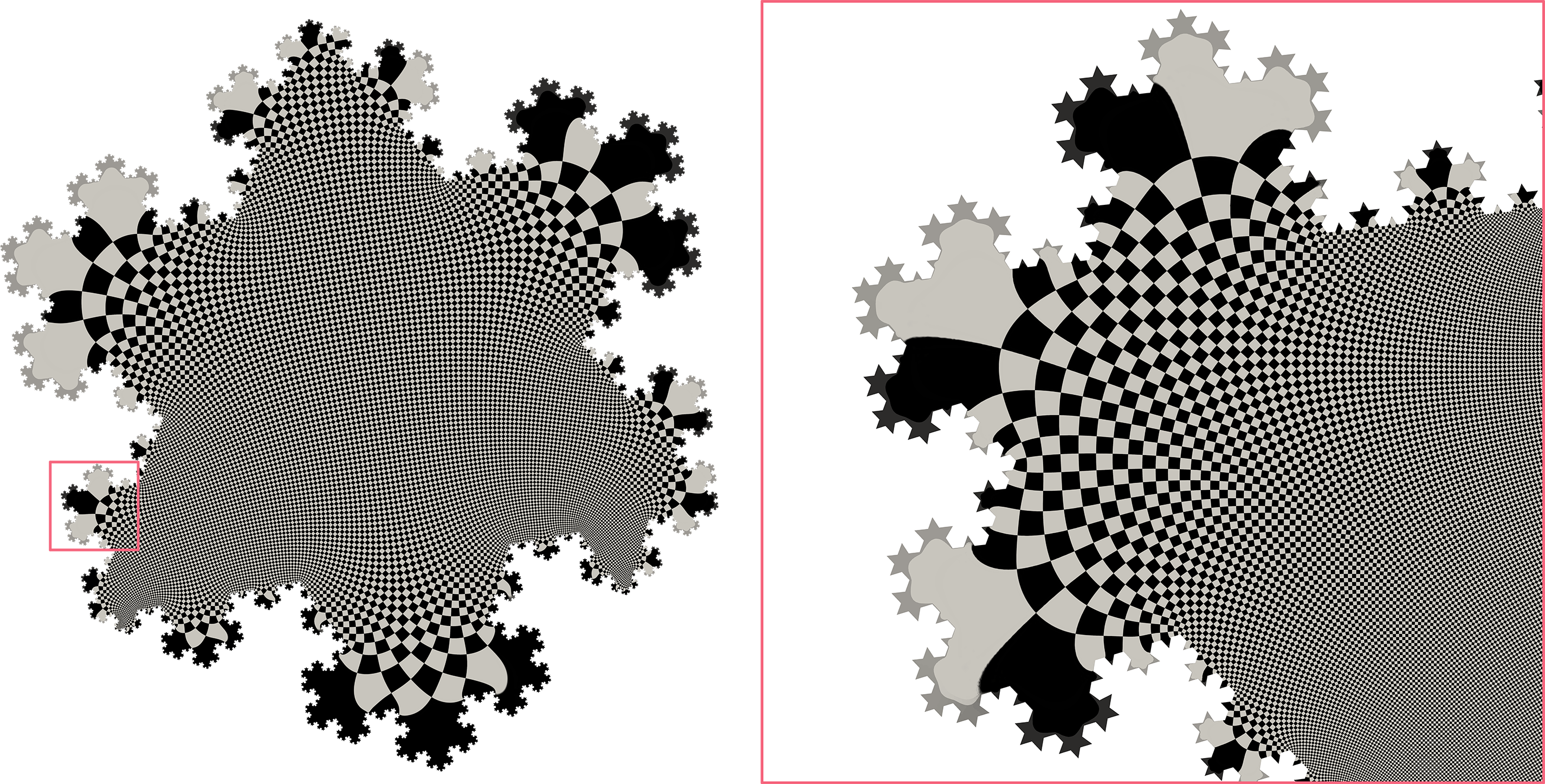}
	\caption{An approximation of the notorious Riemann map from a Koch snow-flake (computed with 6 recursions) to a triangle. The approximation consists roughly 6 million vertex mesh and captures different resolutions of this map as shown in the blow-up on the right.  }
     \label{fig}
     \vspace{-7pt}
\end{figure}

\begin{abstract}
        We consider Riemann mappings from bounded Lipschitz domains in the plane to a triangle. We show that in this case the Riemann mapping has a linear variational principle: it is the minimizer of the Dirichlet energy over an appropriate affine space. By discretizing the variational principle in a natural way we obtain discrete conformal maps which can be computed by solving a sparse linear system.  We show that these discrete conformal maps converge to the Riemann mapping in $H^1$, even for non-Delaunay triangulations. Additionally, for Delaunay triangulations the discrete conformal maps converge uniformly and are known to be bijective. As a consequence we show that the Riemann mapping between two bounded Lipschitz domains can be uniformly approximated  by composing the Riemann mappings between each Lipschitz domain and the triangle.   
\end{abstract}

\section{Introduction}


The Riemann mapping theorem states that there is a biholomorphic  mapping between any two simply connected planar bounded Lipschitz domains\footnote{See \cite{leoni2009first} Def. 12.9 for a definition of Lipschitz domains.} which extends to an homeomorpism between the closures of the domains.  
One of the central themes in the emerging field of Discrete Differential Geometry (DDG) \cite{bobenko2008discrete} aims at developing discrete analogues of conformal mappings. Often the discrete structure at question is a triangulation $\M=(\V,\E,\F)$ of a planar bounded Lipschitz domain $\Omega$, and the question asked is how to place its vertices in the plane or alternatively set its edge lengths to define a discrete analogue of a conformal map into the plane. One important benchmark for discrete conformal mappings is \emph{convergence}. Namely: Does the discrete conformal map converge to the Riemann mapping under refinement of the triangulation $\M$? 

In this paper we construct a \emph{linear} variational principle for the Riemann mapping between a planar bounded Lipschitz domain $\Omega$ and a triangle domain $\T$. We use this principle to devise an algorithm, based on simple piecewise-linear finite-elements, for defining discrete conformal mapping between a simply connected polygonal domain $\Omega$ with arbitrary triangulation $\M$ and a general triangle domain $\T$. This class in particular includes the recent Orbifold-Tutte algorithm \cite{Aigerman:2015:OTE:2816795.2818099} for the case where $\M$ is a Delaunay triangulation and $\T$ is a triangle orbifold (\ie, equilateral or right-angle isosceles). 

The algorithm for computing discrete conformal maps is  \emph{linear} in the sense that it consists of solving a single sparse linear system. We prove that these discrete mappings  converge in the $H^1$ norm to the Riemann mapping $\Phi: \Omega\too \T$ under refinement of the triangulation $\M$. Furthermore, in the case of the Orbifold-Tutte algorithm, where the initial triangulation $\M$ is Delaunay and the triangle $\T$ is an Euclidean Orbifold, the convergence is also uniform over the closure $\bar \Omega$. For two simply-connected polygonal domains $\Omega$, $\Omega'$ with Delaunay triangulations we prove that the composition of the Orbifold-Tutte mappings converge uniformly to the Riemann mapping $\Phi:\Omega\too\Omega'$. 

The linear variational principle is derived from a novel tight linear relaxation of Plateau's problem in the 2D case. Plateau's problem seeks for a surface with minimal area spanning a prescribed curve $\Gamma\subset\Real^d$, $d=2,3$. It is well-known that Plateau's problem can be solved by minimizing the Dirichlet energy of a parameterization $X:B\too \Real^d$, where $B$ is the open unit disc, among all admissible mappings  $X\in \C(\bar B,\Gamma)$ with a (weakly) homeomorphic boundary map $X\vert_{\partial {B}}:\partial B\too \Gamma$ fixing three points on the boundary (see \eg, \cite{dierkes2010minimal}). Formulated this way, Plateau's problem is a variational problem with a convex quadratic energy (Dirichlet) and a non-linear admissible set of functions, $\C(B,\Gamma)$. Therefore, it corresponds to a non-linear partial differential equation in general. When $\Gamma\subset\Real^2$, the unique minimizer of Plateau's variational problem is the Riemann mapping.  We consider a particular instance of Plateau's variational problem: instead of $B$ we consider $\Omega$ as the base domain, and we make a particular choice of $\Gamma = \partial \T \subset \Real^2$ fixing the pre-images of the corners of $\T$. Still, even in this simplified setting, the respective set of admissible mappings, $\C(\Omega,\partial \T)$, for the Plateau's variational problem is non-linear (it is convex, however). We introduce a \emph{relaxation}  of this variational problem by replacing the non-linear admissible set $\C(\Omega,\partial \T)$ with a \emph{linear superset} of admissible mappings $\C^*(\Omega,\partial \T)\supset \C(\Omega,\partial \T)$. Surprisingly, this new variational problem is \emph{tight}, that is, it has a unique solution and this solution is the Riemann mapping $\Phi:\Omega\too \T$. Since this variational problem corresponds to a \emph{linear} partial differential equation, we can employ more or less standard finite-element theory to prove convergence of the algorithm.

The power of this approach is illustrated in Figure~\ref{fig}, where we consider the problem of computing the Riemann mapping of the Koch snow-flake to a triangle. This problem is challenging due to the fragmented nature of the boundary, and requires a high resolution map which would be difficult to achieve using non-linear methods. We obtain such a high resolution map by computing our discrete conformal map from a triangulation of the snow-flake with approximately six million vertices. Solving for the conformal map approximation in this case, using Matlab's linear solver, takes approximately two minutes on an Intel Xeon processor. 

The quality of the discrete conformal mapping is visualized in the standard fashion: A scalar function is defined on the triangle, which represents a black-and-white coloring of a grid. The function is then pulled back by the computed discrete Riemann mapping. Note that the $90^\circ$ angles are preserved by the map. The RHS of figure~\ref{fig} visualizes the high resolution of the map near the boundary.

\section{Related work}
The notion of a discrete conformal mapping of a triangulation $\M$ is a rather well-researched area. It is rich with constructions and algorithms, each with its own definition of discrete conformality, often inspired by some property of smooth conformal mappings. Although we focus here on discrete conformal mappings, we note that there are other numerical algorithms with convergence guarantees to the Riemann mapping based on the Schwartz-Christoffel formula \cite{trefethen1980numerical,natarajan2009numerical}, the zipper algorithm \cite{marshall2007convergence}, polynomial methods \cite{delillo1994accuracy}, and others \cite{henrici1993applied,porter2005history}.

Probably the first discrete conformal mapping is the circle packing introduced in \cite{thurston1979geometry,thurston1985finite}. Circle packing defines a discrete conformal (more generally, analytic) mapping of a triangulation by packing circles with different radii centered at vertices in the plane. These radii can be seen as setting edge lengths in $\M$. Convergence of circle packing to the Riemann mapping was proved in \cite{rodin1987convergence,he1996convergence}. An efficient algorithm for circle packing was developed in \cite{collins2003circle}. A variational principle for circle packing was found in \cite{de1991principe}. Discrete Ricci flow was developed in \cite{chow2003combinatorial} and was shown to converge to a circle packing. Circle patterns \cite{bowers2003planar} generalize circle packings and allow non-trivial intersection of circles; a variational principle for circle patterns was discovered in \cite{bobenko2004variational}. In \cite{luo2004combinatorial} discrete conformality is defined by averaging conformal scales at vertices; in \cite{springborn2008conformal} an explicit variational principle and an efficient algorithm are developed for this equivalence discrete conformality relation. Note that while circle packing has a convex variational principle, it is not linear. Additionally circle packing was shown to converge uniformly on compact subsets of $\Omega$ while our algorithm converges uniformly on all of $\bar \Omega$ and also converges in $H^1$. 

A natural tool, which we also use in this paper,  to handle discrete conformality of triangulations is the Finite Elements Method (FEM) \cite{ciarlet2002finite}. Since Riemann mappings consist of two conjugate harmonic functions, researchers have constructed discrete conformal mappings by  pairs of conjugate discrete harmonic functions defined via the Dirichlet integral  \cite{pinkall1993computing,levy2002least,desbrun2002intrinsic}. These algorithms are linear but do not satisfy any prescribed boundary conditions and are not known to converge to the Riemann mapping. 
Convergence to the Riemann mapping, or more generally the solution of Plateau's problem,  can be obtained by minimizing the Dirichlet energy \cite{tsuchiya1987discrete,tsuchiya2001finite} or a conformal energy \cite{hutchinson1991computing}, while imposing non-linear boundary conditions. Solving these non-convex variational problem is a computational challenge.


 \section{A linear variational principle for Riemann mappings}
We consider a bounded, simply-connected Lipschitz domain $\Omega\subset \Real^2$ with an oriented boundary $\partial\Omega $ and a target triangle domain $\T \subset \Real^2$ with corners $c_1,c_2,c_3\in\Real^2$ positively oriented w.r.t.~$\partial \Omega$. A two dimensional version of Plateau's variational problem is:\begin{subequations}\label{e:plateau}
\begin{align}
\min_{X } & \quad E_D(X) \\
\mathrm{s.t.} & \quad X\in \C(\Omega,\partial \T),
\end{align}
\end{subequations}
where the Dirichlet energy of a map $X(u,v)=(x(u,v), y(u,v))^T:\Omega \too \Real^2$ is defined as $$E_D(X)=\frac{1}{2}\int_\Omega  |X_u|^2+|X_v|^2\, $$
and $\abs{\cdot}$ denotes the standard Euclidean norm of a $2$-vector in $\Real^2$, and the partial derivatives are to be interpreted in the distributional sense. The set of admissible mappings $\C(\Omega,\partial \T)$ is defined as follows: We denote by $H^1(\Omega,\Real^2)$ the Sobolev space of pairs of functions (\ie, $X=(x,y)^T$) with finite Sobolev norm
$$\norm{X}^2=\frac{1}{2}\|X\|^2_{L^2(\Omega)}+E_D(X).$$
In Plateau's problem it is vital to consider boundary values of mappings $X\in H^1(\Omega,\Real^2)$. This is normally done by considering the \emph{trace operator}, $T:H^1(\Omega,\Real^2)\too L^2(\partial \Omega, \Real^2)$, that extends the boundary operator, $TX=X\vert_{\partial\Omega}$, defined on  mappings $X$ which are continuous on $\bar \Omega$, to the entirety of $H^1(\Omega,\Real^2)$  (see \eg, Theorem 1.6.6 in \cite{brenner2007mathematical}). We are now ready to define the set of admissible mappings in Plateau's variational problem \cite{dierkes2010minimal}:
\begin{definition}\label{C}
The admissible function set  $\C(\Omega,\partial \T)$ is defined as follows:
\begin{enumerate}[(i)]
\item \label{e:c_cond1}
$X\in H^1(\Omega,\Real^2) \cap C(\bar \Omega)$.
\item \label{e:c_cond2} $TX$ is a homeomorphism between the boundaries $\partial \Omega$ and $\partial \T$.
\item  \label{e:c_cond3} $TX$ takes three fixed, positively oriented points $p_1,p_2,p_3\in\partial \Omega$ to the corners of the triangle $c_1,c_2,c_3 \in \T$.
\end{enumerate}

\end{definition}

\begin{wrapfigure}[14]{r}{0.3\columnwidth}
\vspace{-13pt}\hspace{-0pt}
\includegraphics[width=0.3\columnwidth]{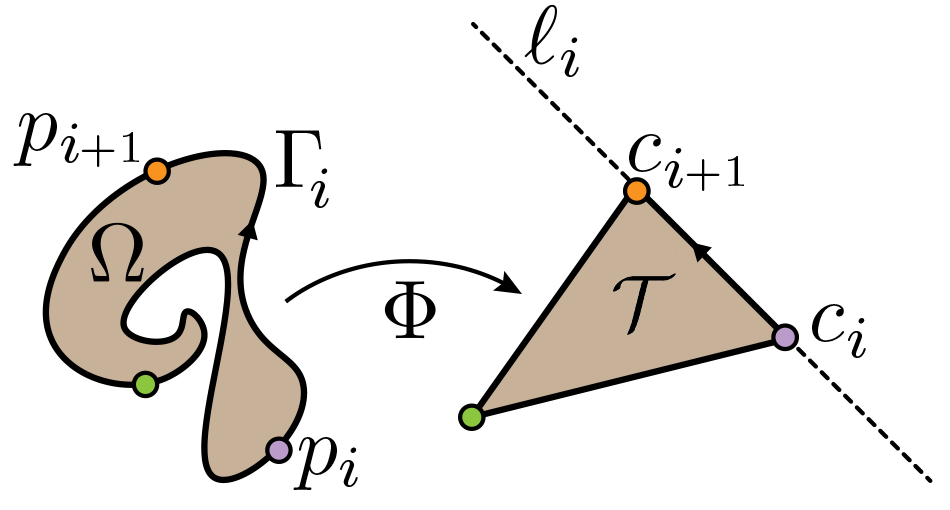}
\caption{The admissible set $\C$ and its relaxation using unordered infinite line constraints, $\C^*$.}\label{fig:setup}
\end{wrapfigure}
The unique minimizer of \eqref{e:plateau} is the  unique Riemann mapping $\Phi:\Omega\too \T$ which satisfies condition \eqref{e:c_cond3}.

We will relax \eqref{e:plateau} by relaxing the homeomorphic condition \eqref{e:c_cond2} in $\C(\Omega,\partial\T)$. Let $\Gamma_i\subset \partial \Omega$, $i=1,2,3$, denote the closed boundary arc connecting $p_i \too p_{i+1}$ (for $i=3$ set $p_4=p_1 $ and $c_4=c_1 $), see Figure \ref{fig:setup}. Then, consider the relaxed admissible mapping space:
\begin{definition}\label{C*}
The relaxed set of admissible mappings, $\C^*(\Omega,\partial\T)$, is defined to be the closure in $H^1(\Omega,\Real^2)$ of the mappings satisfying the following conditions:
\begin{enumerate}[(i)]
\item \label{e:c_cond1}
$X\in H^1(\Omega,\Real^2)$.
\item \label{e:c_cond2} $TX\in C(\partial\Omega,\Real^2)$.
\item  \label{e:c_cond2_relaxed} $a_i^T TX(p) + b_i =0 ,\qquad  \forall p\in \Gamma_i,\quad  \forall i=1,2,3$.
\end{enumerate}
where the infinite line $\ell_i=\set{Z\in\Real^2 \, \vert \,  a_i^T Z+b_i=0}$ supports and infinitely extends the edge $[c_i, c_{i+1}]$ in the triangle $\T$.
 \end{definition}
Figure \ref{fig:setup} illustrates one of the lines $\ell_i$. Note that condition \eqref{e:c_cond2_relaxed} only requires images of points in $\partial\Omega$ to lie on the respective lines $\ell_i$, and nothing prevents the boundary map $TX$ from being non-injective or non-surjective onto $\partial\T$. Also note that we now only require $X$ to be continuous on the boundary.

The first main result of this paper claims that the relaxation
\begin{subequations}\label{e:plateau_relaxed}
\begin{align}
\min_{X } & \quad E_D(X) \\
\mathrm{s.t.} & \quad X\in \C^*(\Omega,\T)
\end{align}
\end{subequations}
is \emph{tight}, that is,
\begin{theorem}\label{thm:main}
The relaxed Plateau's problem \eqref{e:plateau_relaxed} has the Riemann map $\Phi:\Omega \too \T$  satisfying $\Phi(p_i)=c_i$, $i=1,2,3$, as a unique minimizer.
\end{theorem}

\begin{wrapfigure}[4]{r}{0.25\columnwidth}
\vspace{-13pt}\hspace{-5pt}
\includegraphics[width=0.25\columnwidth]{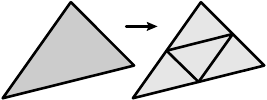}
\end{wrapfigure}
In the second part of the paper we utilize Theorem \ref{thm:main} to show that a piecewise-linear FEM approximation to the minimum of \eqref{e:plateau_relaxed} converges in the $H^1$ norm to the Riemann map under refinement of the triangulation $\M$. Refinement is a sequence of \emph{regular triangulations} $\M_h$ triangulating a polygonal domain $\Omega$ where the maximal edge size $h\too 0$. By regular triangulation we mean that all angles of the triangulations are in some interval $[0+\eps, \pi-\eps]$ for some constant $\eps>0$ (see also \cite{ciarlet2002finite}, p.~124).  One simple subdivision rule that preserves regularity of triangulation is the $1\too 4$ shown in the inset. We further show that if all $\M_h$ are 3-connected and Delaunay (\ie, sum of opposite angles is less than $\pi$) and $\T$ is an Euclidean orbifold, then the convergence is also uniform. Such triangulation families can be computed efficiently by the incremental Delaunay algorithm, for example. 

Let $\Lambda_h\subset H^1(\Omega,\Real^2)$ be the finite dimensional linear space of piecewise-linear continuous functions defined over the triangulation $\M_h$. The Ritz methods for approximating the solution of \eqref{e:plateau_relaxed} is
\begin{subequations}\label{e:plateau_relaxed_fem}
\begin{align}\label{e:plateau_relaxed_fem_E}
\min_{X } & \quad E_D(X) \\ \label{e:plateau_relaxed_fem_admissible}
\mathrm{s.t.} & \quad X\in \C^*(\Omega,\partial \T)\cap \Lambda_h
\end{align}
\end{subequations}
This is a finite-dimensional, linearly-constrained strictly convex quadratic optimization problem (strict convexity follows from Lemma \ref{lem:ed_coer} below) and is uniquely solved via a sparse linear system (the Lagrange multipliers equation). Let $\Psi_h$ denote this solution. We prove:
\begin{theorem}\label{thm:convergence}
Let $\Omega\subset \Real^2$ be a simply connected polygonal domain, $\T\subset \Real^2$ a triangle, $\Phi:\Omega \too \T$ the Riemann mapping satisfying $\Phi(p_i)=c_i$, $i=1,2,3$, and $\M_h$ a sequence of regular triangulations with maximal edge length $h$. Then the solution $\Psi_h$ of \eqref{e:plateau_relaxed_fem} satisfies: $$\lim_{h\too 0}\norm{\Psi_h-\Phi}=0.$$
Furthermore, if all $M_h$ are 3-connected Delaunay and $\T$ is equilateral or right-angled isosceles the convergence is also uniform. 
\end{theorem}
In case the triangle $\T$ is one of the Euclidean orbifolds, that is an equilateral triangle or right-angled isosceles triangle then \eqref{e:plateau_relaxed_fem} is exactly the Orbifold-Tutte algorithm \cite{Aigerman:2015:OTE:2816795.2818099}. If $\M_h$ is Delaunay and $\T$ is an orbifold it is proved in \cite{Aigerman:2015:OTE:2816795.2818099} that $\Psi_h$ is bijective. Since it also converges uniformly by Theorem \ref{thm:convergence} we can approximate the Riemann mapping between two polygons:
\begin{corollary}\label{cor:convergence_2}
Let $\Omega,\Omega'\subset \Real^2$ be two simply connected polygonal domains, $\T\subset \Real^2$ an equilateral or right-angled isosceles triangle, $\Xi:\Omega \too \Omega'$ the unique Riemann mapping satisfying $\Xi(p_i)=p'_i$, $i=1,2,3$, and $\M_h, \M'_h$ two sequences of 3-connected Delaunay regular triangulations of $\Omega,\Omega'$ (resp.)~with maximal edge length $\leq h$, and let $\Psi_h,\Psi'_h$ be the discrete conformal maps from these triangulations to $\T$. Then, $(\Psi'_h)^{-1}\circ\Psi_h$ converges to $\Xi$ uniformly.
\end{corollary}

\section{Proof of tightness (Theorem \ref{thm:main}) }

In this section we prove Theorem \ref{thm:main}, that is show that problem \eqref{e:plateau_relaxed} has a unique solution and this solution is the conformal map $\Phi:\Omega\too \T$.
We start with a Lemma showing that $E_D$ restricted to $\C^*(\Omega,T)$ is \emph{coercive}. Uniqueness of the minimizer will then follow. Let us denote by $\V^*(\Omega,\partial\T)$ the vector space which consists the linear part of $\C^*(\Omega,\partial\T)$. That is,
\begin{equation}
\V^*(\Omega,\partial\T)=\set{X-Y \ \Big\vert \  X,Y\in \C^*(\Omega,\partial\T)}.
\end{equation}

\begin{lemma}\label{lem:ed_coer} The Dirichlet energy satisfies $E_D(X)\geq c\norm{X}^2$ for some constant $c>0$, and  any $X\in \V^*(\Omega,\partial\T)$. The constant $c$ depends only on $\Omega$, $\T$ and the choice of the three points $p_1,p_2,p_3\in\partial\Omega$.  \end{lemma}
\begin{proof}
Let $X\in \V^*(\Omega,\partial\T)$. Since $\norm{X}^2=E_D(X)+\norm{X}^2_{L^2(\Omega,\Real^2)}$ it is enough to show a bound of the form $$\norm{X}^2_{L^2(\Omega,\Real^2)} \leq C E_D(X),$$
for some $C>0$. Denote $\Real^2\ni x=\frac{1}{\area{\Omega}}\int_{\Omega}X$ the average of $X$. Denote $|X|_D=E_D(X)^{1/2}$. We claim that it is sufficient to show that for some $C_1>0$,
        \begin{equation} \label{e:reduced}
        |x| \leq C_1 |X|_D.
        \end{equation}

To see this, use the triangle inequality followed by Poincar\'e inequality (see \cite{leoni2009first}, Theorem 12.23; note that $\Omega$ is in particular a connected extension domain for $H^1(\Omega,\Real)$),
        \begin{align*}
        \norm{X}_{L^2(\Omega,\Real^2)} & \leq \norm{X-x}_{L^2(\Omega,\Real^2)}+\norm{x}_{L^2(\Omega,\Real^2)} \leq C_{\text{Poincar\'e}}|X|_D+ \text{vol}(\Omega)^{\frac{1}{2}}|x|  \\
        & \leq \left( C_{\text{Poincar\'e}}+  C_1\text{vol}(\Omega)^{\frac{1}{2}} \right)|X|_D
        \end{align*}
        where for the last inequality we used \eqref{e:reduced} and in the second to last inequality
        $$\norm{x}_{L^2(\Omega,\Real^2)}=\text{vol}(\Omega)^{\frac{1}{2}}|x|.$$

        We now bound $|x|$ as in \eqref{e:reduced}. Using the trace inequality and the Poincar\'e inequality yet again, we have
\begin{align}\label{e:trace}
\norm{X-x}_{L^2(\partial\Omega,\Real^2)} \leq C_{\text{trace}}\norm{X-x}\leq C_{\text{trace}}(1+C_{\text{Poincar\'{e}}})|X|_D 
        \end{align}
        The square norm of $L^2(\partial\Omega,\Real^2) $ is the sum of the squared norm over each boundary arc $\Gamma_i$, $i=1,2,3$. Denote the length of $\Gamma_i$ by $\length{\Gamma_i}$. So by omitting the last arc we obtain

        \begin{align} \nonumber
        \norm{X-x}_{L^2(\partial\Omega,\Real^2)}^2 & \geq \sum_{i=1}^2 \norm{X-x}_{L^2(\Gamma_i,\Real^2)}^2 \stackrel{(*)}{\geq} \sum_{i=1}^2 \norm{a_i^T(X-x)}_{L^2(\Gamma_i,\Real)}^2  \\
        & \stackrel{(**)}{=} \sum_{i=1}^2 \norm{a_i^Tx}_{L^2(\Gamma_i,\Real)}^2=\sum_{i=1}^2(a_i^Tx)^2 \length{\Gamma_i}   \\ \nonumber
        & = |Ax|^2
        \end{align}

        Where $(*)$ follows from pointwise application of the Cauchy-Schwartz inequality in $\Real^2$, assuming w.l.o.g. that $|a_i|=1$; $(**)$ follows from $X\in\V^*(\Omega,\partial\T)$ and condition \eqref{e:c_cond2_relaxed} in Definition \eqref{C*}, and $A\in \Real^{2\times 2}$ is the invertible matrix with rows $\sqrt{\length{\Gamma_i}}a_i^T$.
        Lastly,
        \begin{align}\label{e:last}
        |x|=|A^{-1}Ax|\leq \norm{A^{-1}}_{2,2}|A x|.
        \end{align}
        Using equations \eqref{e:trace}-\eqref{e:last}, we achieve our goal stated in \eqref{e:reduced}, where 
        $$C_1=\norm{A^{-1}}_{2,2}C_{\text{Trace}}(1+C_{\text{Poincar\'e}}).$$ 
Hence, the constant $C=(C_{\text{Poincar\'e}}+  C_1\text{vol}(\Omega)^{\frac{1}{2}})^2$ and therefore also the constant $c=c(\Omega,\partial\T)$
in the theorem formulation are dependent only on $\Omega$, $\T$ and the choice of three points $p_1,p_2,p_3\in \partial\Omega$.
 \end{proof}

An important consequence of the coercivity of $E_D$ is the uniqueness of the solution of \eqref{e:plateau_relaxed}: 
\begin{lemma}\label{uniqns} The relaxed Plateau's problem \eqref{e:plateau_relaxed} has at most a single solution. \end{lemma}
\begin{proof}
Assume \eqref{e:plateau_relaxed} has two solutions $X,Y\in \C^*(\Omega,\partial\T)$. Restricting $E_D$ to the infinite line $tX+(1-t)Y$, $t\in\Real$, results in a coercive one dimensional quadratic polynomial in $t$ and hence strictly convex. Thus $X=Y$. 
\end{proof}

This Lemma implies that if the conformal map $\Phi$ is a solution to \eqref{e:plateau_relaxed} then it is unique and the relaxation is indeed tight. To show $\Phi$ is a solution we first recall that the Dirichlet energy is an upper bound of the area functional:
\begin{equation}\label{e:DA_ineq}
E_D(X)\geq E_A(X),
\end{equation}
where $$E_A(X)=\int_\Omega \abs{\det \brac{X_u\, X_v}}$$ is the area functional. This inequality can be proved using the inequality $|\det A| \leq \frac{1}{2} |A|_F^2$, where $A\in \Real^{2\times 2}$ and $\abs{\cdot}_F$ is the Frobenious norm of a matrix. When $A$ is a similarity matrix equality holds. For the conformal map $\Phi$, $\brac{\Phi_u , \Phi_v}$ is a similarity matrix everywhere in the open set $\Omega$ and therefore $$E_D(\Phi)=E_A(\Phi)=\area{\T},$$ where $\area{\T}$ denotes the area of the triangle $\T$.
To show that $\Phi$ is a solution to \eqref{e:plateau_relaxed} it is therefore enough to show the following three lemmas. 
\begin{lemma}\label{lem:lower_bound}
Every $X\in \C^*(\Omega,\partial\T)\cap C(\bar \Omega,\Real^2)\cap C^\infty(\Omega,\Real^2)$
satisfies $E_D(X)\geq \area{\T}$.
\end{lemma}
\begin{proof}
Take an arbitrary $X\in \C^*(\Omega,\partial \T)\cap C(\bar \Omega,\Real^2)\cap C^\infty(\Omega,\Real^2)$. We first want to prove that every point $q$ in  $\T$ (the interior of the triangle) has a pre-image $p\in \Omega$.  Assume $q\in\T$, the winding number of $q$ w.r.t.~the restriction of $X$ to $\partial\Omega$ is $w(q,TX)=1$. To see that consider a homeomorphism $Y:\bar \Omega\too \bar \T$ satisfying $Y(p_i)=c_i$, $i=1,2,3$ (\eg, the Riemann mapping). Consider the homotopy $H(\cdot,t)=(1-t)TX(\cdot)+t\, TY(\cdot)$. Note that the image of $H(\cdot,t)$ is contained in $\cup_{i=1}^3 \ell_i$ and since $q$ does not belong to the latter set the winding number $g(t)=w(q,H(\cdot,t))$ is a continuous function of $t$. Since $TY:\partial\Omega\too\partial\T$ is a homeomorphism we have that $g(1)=w(q,H(\cdot,1))=w(q,TY)=1$. We also know that $g(t)\in\Integer$ and therefore $w(q,TX)=w(q,H(\cdot,0))=g(0)=1$. Now to show that $p$ has a pre-image under $X$ we use a mapping degree argument. Assume toward a contradiction that it does not have a pre-image.  Then by the boundary theorem (see \cite{outerelo2009mapping} Proposition 4.4) since $q\notin X(\bar \Omega)$ we have $w(q,TX)=0$, contradiction. The lemma now follows from the following intuitive lemma (Lemma \ref{lem:Area}) which we prove in the appendix.
\end{proof}

\begin{lemma}\label{lem:Area}
	If the image of $X \in C(\bar \Omega,\Real^2)\cap C^\infty(\Omega,\Real^2)$ contains $\T$ then $E_D(X) \geq E_A(\T)\geq |\T| $.
\end{lemma}

\begin{lemma}\label{lem:density}
$\C^*(\Omega,\partial\T)$ equals the closure in $H^1(\Omega,\Real^2)$ of $\C^*(\Omega,\partial\T)\cap C(\bar\Omega,\Real^2)\cap C^\infty(\Omega,\Real^2)$.
\end{lemma}
\begin{proof}
Take $\eps>0$. Let $X\in\C^*(\Omega,\partial\T)$. Then there exist $Y$ satisfying \eqref{e:c_cond1}-\eqref{e:c_cond2_relaxed} in Definition \eqref{C*} so that $\norm{Y-X}\leq\eps/2$. Since $TY$ has continuous representative and has an extension to $\Omega$ that is in $H^1(\Omega,\Real^2)$, we can solve a Dirichlet problem with $TY$ as boundary condition. Let the solution be called $U$. Note that $U\in \C^*(\Omega,\partial\T)\cap C(\Omega,\Real^2)\cap C^\infty(\Omega,\Real^2)$. Furthermore, $T(Y-U)=0$. So we can approximate $Y-U$ with $W\in C_0^\infty(\Omega,\Real^2)$ (\ie, the space of compactly supported smooth functions in $\Omega$), \ie, $\norm{Y-U-W}\leq \eps/2$. Note that $$W+U\in \C^*(\Omega,\partial\T)\cap C(\Omega,\Real^2)\cap C^\infty(\Omega,\Real^2),$$ and
$$\norm{X-(W+U)}\leq \norm{X-Y} + \norm{Y-U-W}\leq \eps.$$
\end{proof}

This concludes the proof of Theorem \ref{thm:main}.

\section{Convergence of Finite-Element Approximations}
We would now like to approximate the Riemann mapping $\Phi:\Omega\too \T$ using the Ritz-Galerkin method. Given a series of regular triangulations $\M_h=(\V_h,\E_h,\F_h)$ of the polygonal domain $\Omega$, with maximal edge length $h\too 0$ we denote by $\Lambda_h\subset \C^*(\Omega,\T)$ the set of continuous piecewise-linear mappings defined over $\M_h$. That is, $\Psi\in \Lambda_h$ is a continuous function that is affine when restricted to each triangle face, $\Psi\vert_{f}$, $f\in \F_h$.  We approximate the Riemann mapping $\Phi$, by solving \eqref{e:plateau_relaxed_fem}. Restricted to the finite-dimensional space $\Lambda_h$, \eqref{e:plateau_relaxed_fem} is a linearly constrained quadratic minimization problem. Indeed, let $\varphi_1,\ldots,\varphi_{\abs{\V_h}}$ denote the standard FEM basis of the continuous piecewise-linear scalar functions over $\M_h$ defined by $\varphi_j(v_k)=\delta_{jk}$, for all $j$ and vertices $v_k\in\V_h$. Then the admissible set \eqref{e:plateau_relaxed_fem_admissible}, $\C^*(\Omega,\partial\T)\cap \Lambda_h$,  can be written as the following affine set in $\Real^{2|\V_h|}$:
\begin{align*}
&X = \sum_j x_j \varphi_j, \qquad \ \set{x_j}_{j=1}^{|\V_h|}\subset \Real^2 \\
&a_i^T  x_k  + b_i = 0, \qquad \forall v_k\in\V_h\cap \Gamma_i, \quad i=1,2,3
\end{align*}
The Dirichlet energy is:
$$E_D(X)=\sum_{kl} x_k x_l W_{kl}, \qquad W_{kl} = \int_\Omega \ip{\nabla \varphi_k, \nabla \varphi_l}$$ 
Lemma \ref{lem:ed_coer} implies that this quadratic form is strictly positive-definite over $\C^*(\Omega,\partial\T)\cap \Lambda_h$ therefore \eqref{e:plateau_relaxed_fem} has a unique solution. This solution is computed by solving the corresponding sparse linear Lagrange multiplier system which can be solved efficiently with, \eg, a direct linear solver. Denoting the solution to \eqref{e:plateau_relaxed_fem} by $\Psi_h$ we would like to prove convergence of $\Psi_h$ to $\Phi$ as the maximal edge length of the triangulations $\M_h$ goes to zero.  For that end, we will use Theorem \ref{thm:main} that identifies $\Phi$ as the unique minimum of the relaxed Plateau's problem \eqref{e:plateau_relaxed}, the coercivity of $E_D$ over $\C^*(\Omega,\partial\T)$, and employ a result from finite element method called C\'{e}a's lemma \cite{brenner2007mathematical} (proved in the appendix for the sake of completeness):

\begin{lemma}\label{lem:cea} \textbf{(C\'{e}a)} Let $\Phi$ be the unique Riemann map in $\C^*(\Omega,\partial\T)$, and $\Psi$ the solution of \eqref{e:plateau_relaxed_fem}. Then, $$\norm{\Phi-\Psi_h} \leq C\norm{\Phi - X}, \quad  \forall  X \in \C^*(\Omega,\partial\T)\cap \Lambda_h,$$ where $C$ is a constant independent of the choice of $\Lambda_h$. \end{lemma}

\begin{proof}[ Proof of Theorem \ref{thm:convergence}]
C\'{e}a's lemma reduces the problem of showing that $\norm{\Phi-\Psi_h}\too 0$ to an approximation problem, \ie, 
\begin{equation}\label{e:approx}
\inf_{X\in \C^*(\Omega,\partial\T)\cap\Lambda_h}\norm{\Phi-X}\too 0,
\end{equation}
as $h\too 0$. We will prove \eqref{e:approx} using the following lemma (proven in the appendix):
\begin{lemma} \label{lem:smooth}
                There is a sequence of functions $\Phi_\epsilon \subseteq  \C^*(\Omega,\partial\T) \cap C^\infty(\bar \Omega,\Real^2)$ which converges to the Riemann mapping $\Phi$ in $H^1(\Omega,\Real^2)$. 
\end{lemma}

The triangle inequality, $\norm{\Phi-X}\leq \norm{\Phi-\Phi_\epsilon} + \norm{\Phi_\epsilon - X}$ and Lemma \ref{lem:smooth} imply that it is enough to approximate $\Phi_\epsilon$ with $X\in  \C^*(\Omega,\partial\T)\cap\Lambda_h$. We take $X_h$ to be the interpolant of $\Phi_\epsilon$, that is the unique function in $\Lambda_h$ which agree with $\Phi_\epsilon$ on the vertices of $\mathcal{T}_h$, \ie, $X_h(v_i)=\Phi_\epsilon(v_i)$, for all $v_i\in \V_h$. Note that $X_h\in \C^*(\Omega,\partial\T)\cap\Lambda_h$. A standard approximation result in the theory of finite elements (\eg, Theorem 4.4.20 in \cite{brenner2007mathematical}) states that since, in particular,  $\Phi_\epsilon \in W_{2}^2(\Omega,\Real^2) $, we have that $\norm{\Phi_\epsilon - X_h}\too 0$ as $h\too 0$. So convergence in $H^1$ norm is proven. That is, $\norm{\Phi-\Psi_h}\too 0$ as $h\too 0$.

To prove uniform convergence, we assume that $\M_h$ are 3-connected Delaunay and that $\T$ is an Euclidean orbifold, namely an equilateral or right-angled isosceles. In this case the Orbifold-Tutte mapping $\Psi_h:\Omega\too \T$ is a homeomorphism \cite{Aigerman:2015:OTE:2816795.2818099}.


Consider $\Phi_h$ to be a solution of the following optimization problem: 
\begin{subequations}\label{e:sample}
\begin{align}\label{e:sample_E}
\min_{X } & \quad E_D(X) \\ \label{e:sample_admissible}
\mathrm{s.t.} & \quad X\in \Lambda_h\\
& \quad X(v)=\Phi(v), \quad \forall v \in \partial\Omega
\end{align}
\end{subequations}
In \cite{ciarlet1973maximum} (see Theorem 2) it is shown that $\Phi_h\too \Phi$ uniformly if $\Phi\in W_p^1(\Omega,\Real^2)$ for some $p>2$. The singularities of the Riemann mapping $\Phi$ are of the form $z^{\alpha}$ (here $z=x+i y$ is a complex variable) where $\alpha\in \Theta\subset (0,2\pi)$, $\Theta$ is a finite set of angles depending on the angles of the polygonal lines $\partial\Omega$ and $\partial\T$. A direct calculation shows that if one takes $p=2+\epsilon$ where $\epsilon>0$ is sufficiently small so that $\alpha>1-\frac{2}{p+\epsilon}$ for all $\alpha\in\Theta$ then $\Phi\in W_p^1(\Omega,\Real^2)$. Therefore, it is enough to show that $\Phi_h-\Psi_h$ converge uniformly to the zero function.

We next want to show that $T\Psi_h=\Psi_h\vert_{\partial\Omega}$ has a subsequence converging uniformly to some continuous limit function $g\in C(\partial\Omega,\Real^2)$. For this part we can assume w.l.o.g.~that $\Omega$ is the unit disk $\B$. If that is not the case we let $\varphi:\B\too\Omega$ be a Riemann mapping and consider $\Psi'_h=\Psi_h\circ \varphi$. Clearly $T\Psi'_h$ converge uniformly to $g\circ\varphi$ iff $T \Psi_h$ converge uniformly to $g$. 

Since $\norm{\Phi-\Psi_h}\too 0$, the Dirichlet energy of $\Psi_h\in H^1(\Omega,\Real^2)\cap C(\bar \Omega,\Real^2)$ is uniformly bounded (remember that the Dirichlet energy is invariant to conformal change of coordinates) and all $\Psi_h$ satisfy $\Psi_h(p_i)=c_i$, $i=1,2,3$. It is known that the Courant-Lebesgue lemma (see \eg, page 257 and Proposition 2, page 263 in \cite{dierkes2010minimal}) implies in this case that $T\Psi_h=\Psi_h\vert_{\partial\Omega}$ has a subsequence converging uniformly to some continuous limit function $g\in C(\partial\Omega,\Real^2)$.

Due to the trace theorem we have that $T\Psi_h$ converges to $T\Phi$ in $L^2(\partial\Omega,\Real^2)$ and therefore $g=T\Phi$. This implies that $T\Psi_h$ converge uniformly to $T\Phi$. Since $\Phi_h$ converge to $\Phi$ uniformly we have that $T\Phi_h$ converge to $T\Phi$ uniformly. 
Since $\M_h$ is Delaunay, $\Psi_h$ and $\Phi_h$ satisfy the discrete maximum principle \cite{floater2003one}. Hence,
$$\norm{\Phi_h-\Psi_h}_{C(\bar{\Omega},\Real^2)}\leq \norm{T\Phi_h-T\Psi_h}_{C(\partial \Omega,\Real^2)} \too  0,$$  
 and since $\Phi_h$ converges uniformly to $\Phi$, this concludes the proof of Theorem \ref{thm:convergence}.
\end{proof}

Uniform convergence easily implies Corollary \ref{cor:convergence_2}:
\begin{proof}[ Proof of Corollary \ref{cor:convergence_2}]
 Using the triangle inequality  we have
 \begin{align*}
 \norm{\Xi - (\Psi'_h)^{-1}\circ\Psi_h}_{C(\bar{\Omega},\Real^2)} & \leq
 \norm{(\Phi')^{-1}\circ\Phi - (\Phi')^{-1}\circ\Psi_h}_{C(\bar{\Omega},\Real^2)}\\
 &+\norm{(\Phi')^{-1}\circ\Psi_h - (\Psi'_h)^{-1}\circ\Psi_h}_{C(\bar{\Omega},\Real^2)}.
 \end{align*}
 Now note that
 $$\norm{(\Phi')^{-1}\circ\Psi_h - (\Psi'_h)^{-1}\circ\Psi_h}_{C(\bar{\Omega},\Real^2)}=\norm{(\Phi')^{-1} - (\Psi'_h)^{-1}}_{C(\bar{\T},\Real^2)}. $$
 Thus the theorem is proven by noting that the convergence of $\Psi_h$ to $\Phi$ and $\Psi_h'$ to $\Phi'$, which was proven in Theorem \ref{thm:convergence}, implies uniform convergence of the inverse function as we ll, and  by noting that convergence is preserved when composing a uniformly converging sequence with a continuous mapping on a compact set from the left.

\end{proof}

        \bibliographystyle{amsplain}
        {
        	\bibliography{plateaus_relaxation}}

\section{Appendix}
\begin{proof}[Proof of Lemma~\ref{lem:Area}]
Define $\T'=\T \setminus S$ where $S$ is the set of critical values of $X$ over $\Omega$. By Sard's theorem $|\T'|=|\T| $. For each $q\in \T'$ we choose some pre-image  $p=p(q)\in \Omega$. By construction $dX_p $ is non-singular and by the inverse mapping theorem, there exists an open neighborhood $p\in \W_q\subset \Omega$ so that $X\vert_{\W_q}$ is a diffeomorphism and is contained in $\T'$. We get an open cover $\T'= \cup_{q \in \T'} X(\W_q)$ and by  Lindel\"{o}f's lemma there exists a countable subcover $\T'= \cup_{k\in \Natural} X(\W_{q_k})$.

We choose a partition of unity $\phi_{k} $ subordinate to this cover, that is  $\phi_k(s)\geq 0$ and $\sum_k \phi_k(s)=1$, for all $s\in \T'$ , and $\mathrm{supp}(\phi_k)\subset X(\W_{{q_k}})$. Now denoting the union of all $\W_{q_k}$ by $\W$ we obtain

\begin{align*}
\area{\T}&=\int_{\T'}1 =\int_{\T'}\sum_{k } \phi_k\\&=  \sum_{k} \int_{ X(\W_{q_{k}})}\phi_k \\ &= \sum_k \int_{\W_{q_k}} \abs{\det \brac{X_u\, X_v}} \phi_k\circ X \\
&\leq \int_{\W} \abs{\det \brac{X_u\, X_v}} \sum_{k} \phi_k\circ X \\
& = \int_{\W} \abs{\det \brac{X_u\, X_v}} \\
& \leq \int_{\Omega} \abs{\det \brac{X_u\, X_v}}  = E_A(X)
\end{align*} 

\end{proof}

\begin{proof}[Proof of Lemma \ref{lem:cea}] Define the bilinear form 
	$$\ip{X,Y}_D=\frac{1}{2}\int_\Omega X_u^T Y_u + X_v^T Y_v.$$
	Note that $\ip{X,X}_D=E_D(X)$ and $\ip{X,X}_D+\frac{1}{2}\int_\Omega |X|^2 = \norm{X}^2$.  Furthermore, Cauchy-Schwarz inequality implies
	\begin{equation}\label{e:cs}
	\ip{X,Y}_D\leq \sqrt{\ip{X,X}_D}\sqrt{\ip{Y,Y}_D}\leq\norm{X}\norm{Y}.
	\end{equation}

	Let $X$ be a minimizer of $E_D$ over some affine space $\Q$. Consider $X(t) = X+tZ$, where $Z$ is a \emph{variation} of $Q$. That is, $Z=Z_1-Z_2$, where $Z_1,Z_2 \in \Q$. $X(t)\in \Q$ for all $t$. The real valued function  $g(t) = E_D(X+tZ)$ is minimized at zero, and the equation $g'(0)=0$ yields
	\begin{align*}
	\ip{X,Z}_D=0.
	\end{align*}
	Since $\Phi$, by Theorem \ref{thm:main}, minimizes $E_D$ over $\C^*(\Omega,\partial\T)$ we have $\ip{\Phi,Z}_D=0$ for all variations $Z\in \V^*(\Omega,\partial\T)$ of $\C^*(\Omega,\partial\T)$. Similarly, by definition $\ip{\Psi_h,Z}_D=0$ for all variations $Z\in \V^*(\Omega,\partial\T)\cap \Lambda_h$ of $\C^*\cap \Lambda_h$. Since $\V^*(\Omega,\partial\T)\cap \Lambda_h \subset \V^*(\Omega,\partial\T)$ we have in particular
	\begin{equation}\label{e:phi_psi}
	\ip{\Phi-\Psi_h,Z}_D=0,
	\end{equation}
	for all variations $Z\in \V^*(\Omega,\partial\T)\cap \Lambda_h$.
	
	Now, for arbitrary $X\in \C^*(\Omega,\partial\T)\cap \V_h$,
	\begin{align*}
	c\norm{\Phi-\Psi_h}^2 &\leq E_D(\Phi-\Psi_h)\\
	&= \ip{\Phi-\Psi_h,\Phi-\Psi_h}_D\\
	&= \ip{\Phi-\Psi_h,\Phi-X}_D \\
	&\stackrel{\eqref{e:cs}}{\leq}\norm{\Phi-\Psi_h}\norm{\Phi-X},
	\end{align*}
	where the first inequality is due to the coerciveness of $E_D$ over $\V^*(\Omega,\partial\T)$ proved in Lemma \ref{lem:ed_coer}. Dividing both sides by $\norm{\Phi-\Psi}$ concludes the proof.
\end{proof}

\begin{proof}[Proof of Lemma~\ref{lem:smooth}]
	Let $\rho: \Real \to [0,1] $ be a smooth function with $\rho(t)=0$ for $t \leq 1$ and $\rho(t)=1$ for $t\geq 2$. Denote by $\rho_\epsilon$ the function $\rho_\epsilon(t)=\rho(t/\epsilon) $. Let $v_i\in\partial\Omega$, $i=1,\ldots,n$, denote the corners of the boundary polygon of $\Omega$. Let $\epsilon_0>0$ be small enough so that the balls $B_{3\epsilon_0}(v_i) $ are all pairwise disjoint, and do not intersect $\partial \Omega$ except for at the two edges which share the vertex $v_i$. For any $\epsilon<\epsilon_0$ define
	\begin{equation}\label{e:PhiEps} \Phi_\epsilon(z)=\sum_{i=1}^n (1-\rho_\epsilon(|z-v_i|)) \Phi(v_i)+ \left(\prod_{i=1}^n\rho_\epsilon(|z-v_i|) \right) \Phi(z)
	\end{equation}
	Note that the restriction of $\Phi_\epsilon$ to $ B_{ \epsilon}(v_i) $ takes the constant value $\Phi(v_i) $, that for   $z \in B_{2 \epsilon}(v_i) $ we have
	$$\Phi_\epsilon(z)=(1-\rho_\epsilon(|z-v_i|))\Phi(v_i)+\rho_\epsilon(|z-v_i|) \Phi(z) $$
	and that $\Phi_\epsilon$ coincides with $\Phi$ on the set $\Omega_\epsilon=\Omega \setminus \cup_{i=1}^n  B_{2 \epsilon}(v_i)$. It is not difficult to verify that the functions $\Phi_\epsilon$ are admissible functions in $\C^*(\Omega,\Real^2)$. To show that  $\Phi_\epsilon$ are in $\C^\infty(\bar \Omega) $, we first extend $\Phi_\epsilon$ to $\U_0=\Omega \cup (\cup_{i=1}^n B_{\epsilon}(v_i)) $ by setting $\Phi_\epsilon=\Phi(v_i) $ on each set $B_{\epsilon}(v_i)$. We then select an open cover $\U_i, i=1,\ldots n $ of  $\partial \Omega \setminus (\cup_{i=1}^n B_{\epsilon/2}(v_i))$, such that each $\U_i$ contains the edge $[v_i,v_{i+1}]$ except for possibly an $\epsilon/2$ radius disc around the vertices, and the sets are all disjoint. On each such set we can use Schwarz's reflection principle to extend $\Phi$ holomorphically to an open set  $\W_i\subset \U_i$, and containing all of $[v_i,v_{i+1}]$ except for an $\epsilon/2$ radius discs around the vertices.  Thus $\Phi_\epsilon$ can be extended to a smooth function on $\U_0\cup \parr{\cup_{i=1}^n \W_i}$ as well using \eqref{e:PhiEps}. Overall we achieve an extension of $\Phi_\epsilon$ to the open set $\U_0\cup \parr{\cup_{i=1}^n \W_i}$ which contains $\bar\Omega$.

	It remains to show that $\Phi_\epsilon$ converges to $\Phi$ in $H^1 $. Due to coercivity, it is sufficient to show  $L^2$ convergence of $h_\epsilon=\partial_k \Phi_\epsilon-\partial_k \Phi$ to zero. For $z \in \Omega_\epsilon$ we have that $h_\epsilon(z)=0$, and thus it is sufficient to show that for all $i$, $\|h_\epsilon \|_{L^2(\Omega \cap B_{2 \epsilon}(v_i),\Real^2)}\rightarrow 0$. We fix some $i$, and  denote $A_{\epsilon}=\Omega \cap B_{2 \epsilon}(v_i) $ and $q(z)=|z-v_i| $.
	For $z \in A_\epsilon \setminus \{v_i\} $ we have that 
	\begin{equation}\label{e:twoSummands}
	h_\epsilon(z)=\underbrace{ \epsilon^{-1} \cdot  \partial_k q(z) \cdot \rho'(\epsilon^{-1}q(z)) \cdot \left[ \Phi(z)-\Phi(v_{i}) \right]}_{\ell_\epsilon}+ \underbrace{ \partial_k \Phi(z) \cdot \left[\rho_\epsilon(q(z))-1 \right] }_{r_\epsilon}
	\end{equation}
	We claim that the $L^2(A_\epsilon,\Real^2)$ norm of $\ell_\epsilon$ and $r_\epsilon$ goes to zero as $\epsilon\too 0$. For $\ell_\epsilon$, note that $\partial_kq $ and $\rho'$ are both bounded uniformly by a constant independent of $\epsilon$, so that there is some $M>0$ such that 
	\begin{align*}
	\int_{A_\epsilon}  |\ell_\epsilon|^2  \leq
	M\epsilon^{-2}\int_{A_\epsilon}  \abs{\Phi(z)-\Phi(v_i) }^2    \leq
	(2\pi \epsilon)^2 \epsilon^{-2} M \sup_{z \in A_\epsilon} \abs{\Phi(z)-\Phi(v_i) }^2 \rightarrow 0
	\end{align*}
	For $r_\epsilon$, note that
	$$\int_{A_\epsilon} |r_\epsilon|^2=\int_{\Omega}|r_\epsilon|^2 \one_{A_\epsilon} \rightarrow 0 $$
	where the convergence follows from the dominated convergence theorem since  $|r_\epsilon|^2 \one_{A_\epsilon}$  converges pointwise to zero almost everywhere and   is bounded by the $L^1$ function $ |\partial_k \Phi |^2$ .     \end{proof}

\paragraph*{Acknowledgements}
This research was supported by the Israel Science Foundation grant No. ISF 1830/17.

\end{document}